%% file: WEBER-ED-delay-arXiv.tex
\pdfoutput=1

\documentclass{article}

\usepackage{savesym}
\usepackage{amsthm,amsmath}
\usepackage{amssymb,latexsym,graphicx}
\usepackage{amscd}
 \usepackage[all,cmtip]{xy}
\usepackage{tikz-cd}

\usepackage{accents}
\usepackage{cite}
\usepackage{mathtools}
\usepackage{stmaryrd} 

\usepackage{calligra}
\DeclareMathAlphabet{\mathcalligra}{T1}{calligra}{m}{n}
\DeclareMathAlphabet{\mathpzc}{OT1}{pzc}{m}{it}
\usepackage{hycolor}
\usepackage{xcolor}
\usepackage[
                       colorlinks=true,
                       linkcolor=black, 
                       citecolor=black, 
                       urlcolor=blue,
                     ]{hyperref}
\usepackage{soul} 

\usepackage{makeidx}

\usepackage[intoc]{nomentbl}
\makenomenclature

\usepackage{stackengine} 
\usepackage{booktabs} 

\input{joa-environments-english}

\input{joa-latex-shortcuts_1902_v1}

\begin{document}
\sloppy

\author{\quad Urs Frauenfelder \quad \qquad\qquad
             Joa Weber
        %
        %
    \\
    Universit\"at Augsburg \qquad\qquad
    UNICAMP
}

\title{The fine structure of Weber's hydrogen atom -- \\
         Bohr-Sommerfeld approach}

\date{\today}

\maketitle 
%


%
%

%





\begin{abstract}
In this paper we determine in second order in the fine structure
constant the energy levels of Weber's Hamiltonian admitting a
quantized torus. Our formula coincides with the formula obtained
by Wesley using the Schr\"odinger equation for Weber's Hamiltonian.
\end{abstract}

\tableofcontents

\section{Introduction}

Although Weber's electrodynamics was highly praised
by Maxwell~\cite[p.\,XI]{Maxwell:1873a}, see
also~\cite[Preface]{Koch-Torres-Assis:1994a},
it was superseded by Maxwell's theory and basically forgotten.
Different from Maxwell's theory, which is a field theory,
Weber's theory is a theory of action-at-a-distance, like Newton's theory.
Weber's force law can be used to explain Amp\`{e}re's law and Faraday's
induction law; see~\cite[Ch.\,4 Ch.\,5]{Koch-Torres-Assis:1994a}.
Different from the Maxwell-Lorentz theory Amp\`{e}re's original law
also predicted transverse Amp\`{e}re forces.
Interesting experiments about the question of existence of 
transverse Amp\`{e}re forces where carried out by
Graneau and Graneau~\cite{Graneau:1996a}.

By quantizing the Coulomb potential Bohr and Sommerfeld,
cf.~\cite[Ch.\,II]{Mehra:1982a},
obtained for the hydrogen atom the following energy levels
\begin{equation}\label{eq:E-levels-Coulomb-Kepler}
     -\frac{1}{2n^2},\quad n\in\N,
\end{equation}
in atomic units.\footnote{
  Differences of these energy levels give rise
  to the \textbf{Rydberg formula}
  $$
     \frac{1}{2n^2}-\frac{1}{2m^2},\quad \text{$n,m\in\N$, $m>n$,}
  $$
  that corresponds to the energy of the emitted photon
  when an electron falls from an excited energy level to a lower one.
  Historically these differences were first measured in spectroscopy.
  In 1885 the swiss school teacher Balmer discovered
  the formula for the $n=2$ series, nowadays referred to as
  the \textbf{Balmer series}; cf.~\cite[p.\,163]{Mehra:1982a}.
}
Later Schr\"odinger interpreted these numbers as eigenvalues of his equation,
cf.\cite{Mehra:1987b}.
By taking account for velocity dependent mass, as suggested by
Einstein's theory of relativity, Sommerfeld obtained a more refined
formula which also takes account of the angular momentum quantum
number $\ell$. Sommerfeld's formula~\cite[\S 2 Eq. (24)]{Sommerfeld:1916a}
is given in second order in the fine structure constant $\alpha$ by
\begin{equation}\label{eq:Sommerfeld}
     -\frac{1}{2n^2}
     -\frac{\alpha^2}{2n^3\ell}
     +{\color{black}\frac{3}{8}}
     \frac{\alpha^2}{n^4},\qquad n\in\N,\quad \ell=1,\dots,n.
\end{equation}
This formula, referred to as the \textbf{fine structure of the
  hydrogen atom}, coincides with the formula one derives
in second order in $\alpha$ from Dirac's
equation as computed by Gordon~\cite{Gordon:1928a}
and~\cite{Darwin:1928a}.
We refer to Schweber~\cite[\S 1.6]{Schweber:1994a}
for the historical context of how Dirac discovered his equation.

The mathematical reason why the 
formula~(\ref{eq:E-levels-Coulomb-Kepler}) of energy levels
for the Coulomb potential of the Kepler problem is degenerate in the sense that it is
independent of the angular momentum quantum number
lies in the fact that the Coulomb problem is super-integrable.
Namely, it is not just rotation invariant, but as well admits
further integrals given by the Runge-Lenz vector.
Dynamically this translates to the fact that for negative energy,
except for collision orbits, all orbits are periodic.
Indeed they are given by Kepler ellipses.
While for velocity dependent mass rotation symmetry is preserved,
the additional symmetry from the Runge Lenz vector is broken.
Dynamically one sees that orbits are in general not any more
periodic, but given by rosettes as illustrated by
Figure~\ref{fig:fig-rosette}.
The situation for Weber's Hamiltonian is quite analogous.
In fact, Bush~\cite{Bush:1926a} pointed out that
the shapes of the orbits in both theories coincide.
However, note even when the shapes of the orbits are the same this does
not mean that their parametrization or their energy coincides.

In this paper we compute the energy levels of the Weber rosettes
using Sommerfeld's method.

\begin{theoremABC}
For Weber's Hamiltonian the fine structure formula for the hydrogen atom
in second order in the fine structure constant $\alpha$ becomes
$$
     -\frac{1}{2n^2}-\frac{\alpha^2}{2n^3\ell}+\frac{1}{2}\frac{\alpha^2}{n^4}.
$$
\end{theoremABC}

\begin{proof}
Equation~(\ref{eq:main}).
\end{proof}

This formula coincides with the formula
Wesley~\cite[Eq.\,(100)]{Wesley:1989a}
obtained using Schr\"odinger's equation.
\newline
Note that the difference to Sommerfeld's formula~(\ref{eq:Sommerfeld})
involves $\alpha^2\approx 10^{-5}$ and
only lies in the term just involving the main quantum number and not
the angular one. Because this term is of much lower order than the
Balmer term, i.e. the first term in~(\ref{eq:Sommerfeld}), it seems
difficult to actually measure the difference.

\subsubsection*{Outlook -- Symplectic topology and non-local Floer
  homology}
Weber's Hamiltonian is related to delayed potentials
as pointed out by Carl Neumann in 1868.
We explain this relation in Appendix~\ref{sec:delayed-potentials}.
The question how to extend Floer theory
to delayed potentials is a topic of active research,
see~\cite{Albers:2018c,Albers:2018a,Albers:2018b}.
In particular, for delayed potentials Floer's equation is not
local and therefore new analytic tools have to be developed
inspired by the recent theory of polyfolds due to
Hofer, Wysocki, and Zehnder~\cite{Hofer:2017a}.
While to our knowledge Weber's Hamiltonian
was so far unknown to the symplectic community,
we hope to open up with this article a new branch of research
in symplectic topology.

To our knowledge so far nobody incorporated the spin of the
electron in Weber's electrodynamics. In fact, it is even an amazing
coincidence that Sommerfeld obtained the same formula
for the fine structure of hydrogen as predicted by Dirac's theory.
It is explained by Keppeler~\cite{Keppeler:2003a,Keppeler:2003b}
that using semiclassical techniques this can be explained,
because the Maslov index and the influence of spin mutually
cancelled out each other which both were not taken into account by
Sommerfeld; neither in our paper.
We expect that the proper incorporation of spin
into Weber's electrodynamics requires techniques
from non-local Floer homology currently under development.

Lamb and Retherford~\cite{Lamb:1947a} discovered 1947 that the
spectrum of hydrogen shows an additional small shift not predicted by
Dirac's theory. This shift nowadays referred to as the \textbf{Lamb shift}
was a major topic in the Shelter Island Conference and the driving
force for the development of Quantum Field Theory and Renormalization
Theory; see~\cite[\S 4 \S 5]{Schweber:1994a}.
An intuitive explanation of the Lamb shift is that vacuum fluctuations
cause a small correction of the potential energy close to the nucleus
and this small correction then leads to a shift in the spectrum of the
hydrogen atom.

In Appendix~\ref{sec:delayed-potentials}
we explain following Neumann how a retarded Coulomb potential
when Taylor approximated up to second order in the fine structure constant
leads to Weber's Hamiltonian.
Higher order perturbations lead to perturbations of Weber's
Hamiltonian which are most strongly felt close to the nucleus.
Whether there is a relation between these higher order perturbations
and vacuum fluctuations is an important topic in the non-local
Floer homology under development and its interaction
with the semiclassical approach.

\vspace{.2cm}
\textit{Acknowledgements.}
The authors would like to thank sincerely to Andr\'{e}
Koch Torres Assis for many useful conversations about
Weber's electrodynamics.
The paper profited a lot from discussions with Peter Albers
and Felix Schlenk about delay equations whom we would like
to thank sincerely as well.
This article was written during the stay of the first author
at the Universidade Estadual de Campinas (UNICAMP) whom he
would like to thank for hospitality.
His visit was supported by UNICAMP and Funda\c{c}\~{a}o de Amparo
\`{a} Pesquisa do Estado de S\~{a}o Paulo (FAPESP),
processo $\mathrm{n}^{\rm o}$ 2017/19725-6.

\section{Weber's Hamiltonian}

In this article we use atomic units to describe a model for the
hydrogen atom. There are four \textbf{atomic units} which are unity:
The electron mass $m_e=1$, the elementary charge $e=1$,
the reduced Planck constant $\hbar=h/2\pi=1$, and the Coulomb
force constant $k_0=1/(4\pi\eps_0)=1$.
In particular, Coulomb's Lagrangian and Hamiltonian
for an electron ($-e=-1$) attracted by a proton are given by
$$
     L=\frac12\abs{v}^2+\frac{1}{\abs{q}},\qquad
     H=\frac12\abs{p}^2-\frac{1}{\abs{q}},
$$
respectively. The speed of light is given in atomic units by
$c=\frac{1}{\alpha}\approx 137$ where $\alpha$ is
Sommerfeld's \textbf{fine structure constant}.
In his work Weber used a different constant, namely
$$
     c_{\rm W}=\sqrt{2} c.
$$
In his famous experiment in 1856 he measured
this constant together with Kohlrausch.
This experiment was later crucial for Maxwell, because it indicated a
strong relationship between electrodynamics and light~\cite{Weber:1857a}.

The hydrogen atom consists of a (heavy) proton and a (light) electron.
We just consider the planar case and suppose that the proton sits at the
origin of $\R^2$.
Polar coordinates $(r,\phi)\in\R^2_\times:=\R^2\setminus\{0\}$
provide the coordinates $(r,\phi,v_r,v_\phi)\in T\R^2_\times$
and $(r,\phi,p_r,p_\phi)\in T^*\R^2_\times$.
Consider the Lagrangian function
\begin{equation}\label{eq:L_W}
     \LW(r,\phi,v_r,v_\phi)
     =\underbrace{\frac{1}{2}( v_r^2+r^2 v_\phi^2)}_{=:\Tflat}
     +\underbrace{\frac{1}{r}\bigg(1+\frac{
         v_r^2}{2c^2}\bigg)}_{=:-S} 
\end{equation}
Here $\Tflat$ is just the kinetic energy in polar coordinates,
while $S$ is referred in~\cite{Koch-Torres-Assis:1994a}
as the 'Lagrangian energy'. 
Note that $\LW$ differs from the Coulomb Lagrangian $L$
by the additional term $v_r^2/2c^2r$.
This Lagrangian was introduced by Carl Neumann~\cite{Neumann:1868a}
in 1868. 
We refer to $S$ as \textbf{Neumann's potential function}.
Its Euler Lagrange equation are precisely the equations
studied by Wilhelm Weber~\cite{Weber:1846a} twenty years earlier in 1846.
In Appendix~\ref{sec:delayed-potentials}
we explain how Neumann's potential function can be obtained
as Taylor approximation of a retarded functional.

Observe that the $S$ term in $\LW$ depends on the velocity. 
Historically this led to a lot of confusion
and so Helmholtz~\cite{Helmholtz:1847a} and
Maxwell~\cite{Maxwell:1855a} doubted for a long time
that Weber's force law complies with conservation of energy;
cf.~\cite[\S 3.6]{Koch-Torres-Assis:1994a}.
In~\cite[\S 384/385]{Thomson:1867a} Weber's theory was even classified
under the theories \textit{``pernicious rather than useful.''}
\\
Although in 1871 Weber explained in detail that his force law
satisfies the principle of preservation of energy, his
article~\cite{Weber:1871a} was not mentioned in the
translation~\cite{Thomson:1871a}
of Thompson and Tait's book by Helmholtz and Wertheim
which infuriated Z\"ollner~\cite[Vorrede]{Zollner:1872a}
\textit{``Ich wage es zuversichtlich zu behaupten, dass in der ganzen
deutschen Literatur nicht ein einziges Lehrbuch anzutreffen sein wird,
welches wie jener ber\"uchtigte \S 385 des Werkes von
{\sc Thomson} und {\sc Tait} auf dem engen Raume von nur dreissig
Zeilen eine solche F\"ulle von \emph{absolutem} Nonsens
enth\"alt.''} 
In the new edition of Thompson and
Tait~\cite{Thomson-Lord-Kelvin:1879a} the infamous $\S 385$
has disappeared.

A symplectic way to see that the disputed preservation of energy holds
for Weber's force law is to rearrange in~(\ref{eq:L_W}) the brackets to obtain
\begin{eqnarray*}
\LW(r,\phi,v_r,v_\phi)
&=&\frac{1}{2}\bigg(1+\frac{1}{c^2r}\bigg) v_r^2
    +\frac12  r^2 v_\phi^2+\frac{1}{r}
\\
&=&\frac{1}{2}\left( g_{rr}\, v_r^2
    +g_{\phi\phi}\, v_\phi^2\right)+\frac{1}{r}
\\
&=&
     T-V.
\end{eqnarray*}
The first term in the sum can be interpreted as kinetic energy with
respect to a non-flat Riemannian metric $g$ on $\R^2_\times$,
while the second term is minus the (velocity independent) Coulomb
potential. In Appendix~\ref{sec:proton-proton}
we explain how in the case of two protons the metric becomes
singular at Weber's critical radius $\rho$. Outside this critical radius the
metric is Riemannian, while inside it is Minkowski.

Legendre transformation $\Ll:T\R^2_\times\to T^*\R^2_\times$
of the mechanical Lagrangian $L=T-V$ yields the mechanical Hamiltonian
\begin{eqnarray*}
\HW(r,\phi,p_r,p_\phi)
     &=&T^*+V
\\
&=&\frac{1}{2}\left(\frac{1}{g_{rr}}\,
    p_r^2+\frac{1}{g_{\phi\phi}}\,p_\phi^2\right)-\frac{1}{r}\\
&=&\frac{1}{2}\left(\frac{c^2r}{c^2r+1}\,p_r^2
     +\frac{1}{r^2}\,p_\phi^2\right)-\frac{1}{r}.
\end{eqnarray*}
But this is an autonomous Hamiltonian and any such is preserved along its
flow; see e.g.~\cite[p.\,99]{McDuff:2017b}. This explains preservation of energy.

With the help of the fine structure constant $\alpha=\frac{1}{c}$ the
Hamiltonian reads
\begin{equation}\label{eq:HW}
     \HW(r,{\color{brown}\phi},p_r,p_\phi)
     =\frac{1}{2}\frac{r}{r+\alpha^2}\, p_r^2+\frac{1}{2r^2}\, p_\phi^2-\frac{1}{r}.
\end{equation}
We refer to $\HW$ as \textbf{Weber's Hamiltonian}.

\section{Quantized tori}

Weber's Hamiltonian is completely integrable
in the planar case under consideration.
Indeed it is rotational invariant and therefore
angular momentum commutes with $\HW$.

On closed symplectic manifolds the \Arnold-Liouville
theorem, see e.g.~\cite[App.\,A.2]{Hofer:2011a},
tells us how the manifold gets foliated by invariant tori.
Even below the escape threshold, i.e. for negative energy, in the case
at hand the tori might have some holes, because of collisions.
\\
For the Kepler Hamiltonian collisions can be regularized.
There exist different regularizations in the literature.
For example, Moser showed in~\cite{Moser:1970a} that for
negative energies $E<0$ the Kepler flow after regularization
can be identified with the geodesic flow on a $2$-dimensional sphere.
For zero energy $E=0$ Kepler flow after regularization becomes
identified with geodesic flow on the euclidean plane,
for positive energy $E>0$ with the geodesic flow on hyperbolic space,
as shown by Belbruno~\cite{Belbruno:1977a} and Osipov~\cite{Osipov:1977a};
see also Milnor~\cite{Milnor:1983a}.
Historically even older is the regularization
by Goursat~\cite{Goursat:1889a} which was 
rediscovered independently by Levi-Civita~\cite{Levi-Civita:1920a}
and is nowadays referred to as Levi-Civita
regularization.
Different from Moser, the Levi-Civita regularization is $2:1$
and transforms the Kepler flow for negative energy $E<0$
to the flow of two uncoupled harmonic oscillators.

To our knowledge so far nobody studied regularization
for the Weber Hamiltonian and this would be an interesting project.
Nevertheless, one can easily see how the invariant tori look like.
Indeed for negative energy usual orbits, apart from circle orbits and
collisions, are given by \emph{rosettes}; see
Figure~\ref{fig:fig-rosette}.
\begin{figure}
  \centering
  \includegraphics
                             [height=4cm]
                             {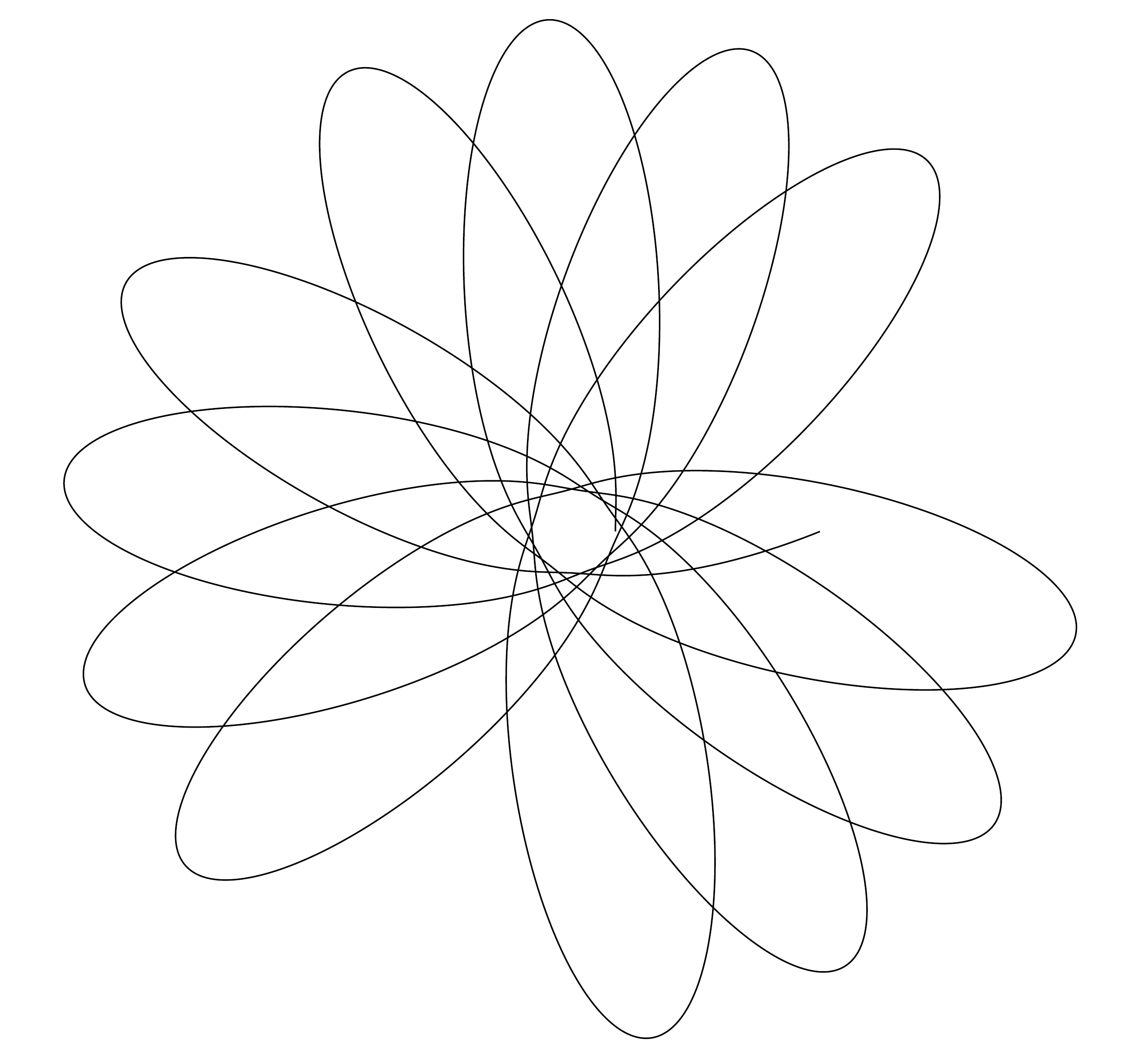}
  \caption{A rosette ($r=1+\kappa \cos \gamma\phi$, eccentricity
    $0<\kappa<1$)}
  \label{fig:fig-rosette}
\end{figure}
Different from the Kepler problem where these orbits are given by
Kepler ellipses, the rosettes don't need to be closed and they show a
perihel shift; see references
in~\cite[p.\,56 footnote\,37]{Wolfschmidt:2018a}.
In the case at hand,  where the central body
is interpreted as a proton, let us replace the expression
'perihel shift' by \textbf{periproton shift}.
If the periproton shift is a rational multiple of $2\pi$, then the
rosette finally closes and we obtain a periodic orbit.
Hence by rotation invariance of $\HW$
we obtain a circle family of periodic orbits
by rotating our closed rosette.
But a circle times a circle is a 2-dimensional torus.
In this case the flow on the invariant torus is rational.

If the periproton shift is irrational, the rosette is not closed and
we obtain the invariant torus by looking at the closure of the rosette.
In this case the flow on the invariant torus is irrational.

A conceptual explanation what Bohr-Sommerfeld quantization
of a completely integrable system is was given by
Einstein~\cite{Einstein:1917a} in 1917; see
also~\cite[\S 14.1]{Gutzwiller:1990a}.
The invariant tori are Lagrangian. Hence if we consider the
restriction of the Liouville 1-form to the torus $T$ it is closed
and therefore it defines a class $[\lambdacan]\in\Ho^1(T;\R)$ in the first
cohomology of the torus.
If this class is integer-valued, then we call the torus a
\textbf{quantized torus}.
For some unknown reason the electron just likes to stay on quantized
tori. Emission occurs if the electron jumps from one quantized torus
to another one.
In this case the frequency we observe is given by $\frac{1}{2\pi}$
times the energy difference of the two energy levels where the
quantized tori lie.

Therefore to understand the spectrum of the electron we have to figure
out which energy level contain a quantized torus.
How to do this in practice is the content of
Sommerfeld's book~\cite{Sommerfeld:1916a}.

In the next section we explain how one apply Sommerfeld's
calculations to the Weber Hamiltonian.

\section{Bohr-Sommerfeld quantization}

We quantize the Hamiltonian according to the rules of Bohr and
Sommerfeld. Note that $\HW$ does not depend on ${\color{brown}\phi}$,
so $p_\phi$ is a preserved quantity that corresponds to
angular momentum. According to Bohr~\cite{Bohr:1913a} angular
momentum has to be quantized: The
\textbf{angular momentum quantum number} is
\begin{equation}\label{eq:quant-p_phi}
     \frac{1}{2\pi}\int_0^{2\pi} p_\phi\, d\phi =p_\phi=:\ell\in\N.
\end{equation}
Here the first identity holds by preservation of angular momentum.
Later Sommerfeld referred to $\ell$ as the azimuthal quantum number
which in his notation was called $n$.
Originally Bohr just considered circular orbits and therefore the
azimutal quantum number was enough.
In contrast, Sommerfeld~\cite{Sommerfeld:1916a} allowed as well more
general orbits and imposed a quantization condition on $p_r$ as well, namely
$$
     \frac{1}{2\pi}\int p_r\, dr =:n_r\in\N_0.
$$
This integral has to be interpreted as follows.
Note that because of rotational invariance (independence of $\phi$) of
the Weber Hamiltonian $\HW$, orbits in the configuration space are
given by rosettes, i.e. the $r$ variable is periodic in time,
oscillating between the periproton $r_{\rm min}$ and the apoproton
$r_{\rm max}$ (closest/farthest point from proton). 
Therefore 
$$
     \int p_r\, dr
     =2 \int_{r_{\rm min}}^{r_{\rm max}} p_r\, dr.
$$
How to interpret and calculate $\int_{r_{\rm min}}^{r_{\rm max}} p_r\, dr$
is illustrated by~\cite[p.\,478 Fig.\,101]{Sommerfeld:1921a}.
Using~(\ref{eq:HW}) and~(\ref{eq:quant-p_phi}) we compute for $p_r$
the formula
\begin{eqnarray*}
p_r&=&\sqrt{2\bigg(1+\frac{\alpha^2}{r}\bigg)\bigg(\HW +\frac{1}{r}-\frac{\ell^2}{2r^2}\bigg)}\\
&=&\sqrt{2\HW +\frac{2+2\alpha^2 \HW }{r}-\frac{\ell^2-2\alpha^2}{r^2}-\frac{\ell^2\alpha^2}{r^3}}.
\end{eqnarray*}
We abbreviate
\begin{equation}\label{eq:ABC}
     A=2\HW ,\quad B=1+\alpha^2 \HW , \quad C=-\ell^2+2\alpha^2<0, \quad
     D_1=-\ell^2 \alpha^2
\end{equation}
so that $p_r$ becomes
$$
  p_r=\sqrt{A+\frac{2B}{r}+\frac{C}{r^2}+\frac{D_1}{r^3}}.
$$
As calculated by Sommerfeld~\cite[p.480\,(16)]{Sommerfeld:1921a}
the integral $n_r$ is given by
$$
     n_r=-i\left(\sqrt{C}-\frac{B}{\sqrt{A}}
     -\frac{BD_1}{2C\sqrt{C}}\right)
$$
where as explained in~\cite[p.479]{Sommerfeld:1921a}
the square root of $C$ has to be taken \emph{negative} imaginary,
whereas the one of A \emph{positive}.
Plugging in~(\ref{eq:ABC}) we obtain
\begin{eqnarray*}
 n_r &=&-i\bigg(-i\sqrt{\ell^2-2\alpha^2}-\frac{1+\alpha^2 \HW }{\sqrt{2\HW }}+\frac{(1+\alpha^2 \HW )\ell^2\alpha^2}{2i(\ell^2-2\alpha^2)^{3/2}}\bigg)\\
&=&-\sqrt{\ell^2-2\alpha^2}+i\frac{1+\alpha^2 \HW }{\sqrt{2\HW }}-\frac{(1+\alpha^2 \HW )\ell^2\alpha^2}{2(\ell^2-2\alpha^2)^{3/2}}.
\end{eqnarray*}
Taking the two imaginary terms to the left hand side shows that
$$
     n_r +\sqrt{\ell^2-2\alpha^2}+\frac{(1+\alpha^2 \HW
     )\ell^2\alpha^2}{2(\ell^2-2\alpha^2)^{3/2}}=i\frac{1+\alpha^2 \HW
     }{\sqrt{2\HW }}.
$$
By Taylor expansion of the left hand side in $(\alpha^2)$ up to first order we get
\begin{equation}
\begin{split}
     n_r +\sqrt{\ell^2-2\alpha^2}+\frac{(1+\alpha^2 \HW
     )\ell^2\alpha^2}{2(\ell^2-2\alpha^2)^{3/2}}
   &\approx n_r +\ell-\frac{\alpha^2}{\ell}+\frac{\ell^2
     \alpha^2}{2\ell^3}\\
   &= (n_r+\ell)-\frac{\alpha^2}{2\ell}.
\end{split}
\end{equation}
Plugging this formula into the previous formula and taking squares
we obtain to first order in $(\alpha^2)$ the approximation
$$
     ( n_r +\ell)^2-\frac{\alpha^2(\ell+
       n_r)}{\ell}\approx-\frac{1+2\alpha^2 \HW }{2 \HW }
     =-\frac{1}{2\HW}-\alpha^2.
$$
Rearranging we get
$$
     \frac{-1}{2\HW }\approx( n_r +\ell)^2-\frac{\alpha^2  n_r }{\ell}
$$
and therefore, expanding again up to first order in $(\alpha^2)$, we obtain
\begin{eqnarray*}
\HW &\approx&-\frac{1}{2\big(( n_r +\ell)^2-\frac{\alpha^2  n_r
              }{\ell}\big)}\approx -\frac{1}{2( n_r
              +\ell)^2}-\frac{\alpha^2  n_r }{2( n_r +\ell)^4 \ell}.
\end{eqnarray*}
By introducing the \textbf{main quantum number}
$$
     n:=n_r+\ell\in\N
$$
the previous formula becomes
\begin{equation}\label{eq:main}
     \HW\approx
     -\frac{1}{2n^2}-\frac{\alpha^2}{2n^3\ell}+\frac{\alpha^2}{2n^4}.
\end{equation}
The corresponding formula for Sommerfeld's relativistic Hamiltonian is
$$
     H_{\rm S}\approx
     -\frac{1}{2n^2}-\frac{\alpha^2}{2n^3\ell}+\frac{3\alpha^2}{8n^4}.
$$

\appendix
\section{Weber's Lagrangian and delayed potentials}
\label{sec:delayed-potentials}

In several works Neumann treated the connection between Weber's dynamics
and delayed potentials, see~\cite{Neumann:1868a,Neumann:1868b}
and~\cite[Ch.\,8]{Neumann:1896a}.
Neumann explained how Weber's potential function is related
to Hamilton's principle which in~\cite{Neumann:1868b} he called
\textit{``norma suprema et sacrosancta, nullis exceptionibus obvia''}.

Strictly speaking, a delay potential only makes sense
for loops and not for chords. Hence we abbreviate by
$\mathcal{L}=C^\infty\big(S^1, \mathbb{R}^2_\times\big)$
the free loop space on the punctured plane
$\R^2_\times:=\mathbb{R}^2 \setminus \{0\}$.
For a potential 
$V \in C^\infty\big(\mathbb{R}^2_\times,\mathbb{R}\big)$
and a constant  $c _{\rm W}>0$ we define three functions
$$
     \Ss_{\mathrm{kin}}, \Ss_{\mathrm{pot}}, \Ss \colon \mathcal{L} \to
     \mathbb{R}
$$
by $\Ss=\Ss_{\mathrm{kin}}-\Ss_{\mathrm{pot}}$ and by
$$
     \Ss_{\mathrm{kin}}(q):=\frac{1}{2}\int_0^1|q'(t)|^2 dt,\qquad
     \Ss_{\mathrm{pot}}(q):=\int_0^1 V\Big(q\Big(t-\tfrac{|q(t)|}{c
       _{\rm W}}\Big)\Big)dt.
$$
Physically this means that the potential energy is evaluated
at a retarded time.
Namely, the position of the proton at the origin has to be
transmitted to the electron at speed $c_{\rm W}$.
It is a strange fact that to obtain Weber's force
this transmission velocity is given by the Weber constant $c_{\rm W}$
which, as measured by Weber and Kohlrausch~\cite{Weber:1857a},
equals $\sqrt{2} c$ where $c$ is the speed of light.

We assume that $V$ only depends on the radial coordinate
$V(q)=V(|q|)=V(r)$. Setting $r(t):=\abs{q(t)}$ the functional
$\Ss_{\mathrm{pot}}$ becomes a function of $r=r(t)$, still denoted by
$$
     \Ss_{\mathrm{pot}}(r)=\int_0^1 V\Big(r\Big(t-\tfrac{r(t)}{c_{\rm
     W}}\Big)\Big)dt.
$$
We further abbreviate
$$
     a=\frac{1}{c_{\rm W}}=\frac{\alpha}{\sqrt{2}},\qquad
     c_{\rm W}=\sqrt{2} c,
$$
where $\alpha$ is the fine structure constant.

Setting
$V_r( a,t):=V\big(r\big(t- a r(t)\big)\big)$
we obtain for the partial derivatives in $a$ the formulas
$$\frac{\partial}{\partial a}V_r( a,t)=-V'\big(r\big(t- a
r(t)\big)\big)r'\big(t- a r(t)\big)r(t)$$
and
\begin{equation*}
\begin{split}
     \frac{\partial^2}{\partial a^2}V_r( a,t)
   &=V''\big(r\big(t- a r(t)\big)\big)r'\big(t- a r(t)\big)^2r(t)^2
\\
   &\quad+V'\big(r\big(t- a r(t)\big)\big)r''\big(t- a r(t)\big)r(t)^2.
\end{split}
\end{equation*}
In particular, at $a=0$ we get
$$\frac{\partial}{\partial a}V_r(0,t)=-V'(r(t))r'(t)r(t)$$
and
$$\frac{\partial^2}{\partial a^2}V_r(0,t)=V''(r(t))r'(t)^2r(t)^2+
V'(r(t))r''(t)r(t)^2.$$
We define
$$
     \Ss^k_{\mathrm{pot}}(r):=\frac{1}{k!}\int_0^1
\frac{\partial^k}{\partial  a^k} V_r(0,t)\, dt,\qquad k\in\N_0.
$$
For $k=0$ this is the unretarded action functional
$$
     \Ss^0_{\mathrm{pot}}(r)=\int_0^1 V(r(t))\, dt.
$$
To see that $\Ss^1_{\mathrm{pot}}\equiv 0$ vanishes identically
choose a primitive $F$ of the function in one variable $r\mapsto V'(r)
r$, that is $F'(r)=V'(r) r$. Indeed, we get that
\begin{equation*}
\begin{split}
\Ss^1_{\mathrm{pot}}(r)&=-\int_0^1 V'(r(t))r'(t)r(t)\, dt\\
&=-\int_0^1 F'(r(t))r'(t)\, dt\\
&=-\int_0^1
\frac{d}{dt}F(r(t))\, dt\\
&=0.
\end{split}
\end{equation*}
The last equation follows, because $r(t)=\abs{q(t)}$ is periodic.
Using integration by parts for the second term in the sum we get that
\begin{eqnarray*}
\Ss^2_{\mathrm{pot}}(r)&=&\frac{1}{2}\int_0^1 \Big(V''(r(t))r'(t)^2r(t)^2+
\underbrace{r''(t)}_{u'}\underbrace{V'(r(t))r(t)^2}_{v}\Big)dt\\
&=&\frac{1}{2}\int_0^1 V''(r(t))r'(t)^2r(t)^2\,dt\\
& &-\frac{1}{2}\int_0^1
\underbrace{r'(t)}_{u}\underbrace{
\left(V''(r(t))r'(t)r(t)^2+2V'(r(t))r'(t) r(t)\right)}_{v'}dt\\
&=&-\int_0^1 V'(r(t))r'(t)^2 r(t) \, dt.
\end{eqnarray*}
For the Coulomb potential
$$
     V(r)=-\frac{1}{r},\qquad V'(r)=\frac{1}{r^2},
$$
this simplifies to
$$
     \Ss^2_{\mathrm{pot}}(r)=-\int_0^1 \frac{r'(t)^2}{r(t)}\,dt.
$$
Hence Taylor approximation of $\Ss_{\mathrm{pot}}$ to second order
in $a=\frac{1}{c_{\rm W}}$ leads to
\begin{equation*}
\begin{split}
     \Ss^0_{\mathrm{pot}}(r)+\frac{1}{c _{\rm
     W}}\Ss^1_{\mathrm{pot}}(r)+\frac{1}{c_{\rm
     W}^2}\Ss^2_{\mathrm{pot}}(r)
   &=-\int_0^1 \frac{1}{r(t)}\bigg(1+\frac{r'(t)^2}{c _{\rm
         W}^2}\bigg) dt\\
   &=\int_0^1 S(r(t))\,dt
\end{split}
\end{equation*}
where
$$
      S(r)=-\frac{1}{r}\bigg(1+\frac{r'^2}{2c^2}\bigg)
$$
is Neumann's potential function,
see~(\ref{eq:L_W}), and $2c^2=c _{\rm W}^2$.

\section{Proton-proton system - Minkowski metric}
\label{sec:proton-proton}

For a positive charge influenced by the proton
the Weber force exhibits fascinating properties as well.
For simplicity suppose both charges are protons
and set their mass equal to one.
In this case the Lagrangian function is given by
\begin{equation*}\label{eq:L_W-pp}
     \LW(r,\phi,v_r,v_\phi)
     =\frac{1}{2}( v_r^2+r^2 v_\phi^2)
     {\color{magenta}-}\frac{1}{r}\bigg(1+\frac{
         v_r^2}{2c^2}\bigg).
\end{equation*}
Changing brackets
\begin{equation*}
\begin{split}
\LW(r,\phi,v_r,v_\phi)
&=\frac{1}{2}\bigg(1 {\color{magenta}-}\frac{1}{c^2r}\bigg) v_r^2
    +\frac12  r^2 v_\phi^2 {\color{magenta}-}\frac{1}{r}
\\
&=\frac{1}{2}\left( g_{rr}\, v_r^2
    +g_{\phi\phi}\, v_\phi^2\right)-\frac{1}{r}.
\end{split}
\end{equation*}
Now the metric gets singular at \textbf{Weber's critical radius}
$$
     \rho:=\frac{1}{c^2}=\alpha^2.
$$
Outside Weber's critical radius the metric is Riemannian,
while inside it is Minkowski. An interesting aspect of Weber's
critical radius is that while outside Weber's critical radius
the force is repulsing -- inside it is \emph{attracting}!
This led Weber to predict -- 40 years before Rutherford's
experiments -- an atom consisting of a nucleus build
of particles of the same charge together with particles
of the opposite charge moving around the nucleus
like planets.
For more informations about Weber's planetary model of the atom
see~\cite{Wolfschmidt:2011a,Wolfschmidt:2018a}.

\bibliographystyle{alpha}
\addcontentsline{toc}{section}{References}
\bibliography{$HOME/Dropbox/0-Libraries+app-data/Bibdesk-BibFiles/library_math}{}

%


\end{document}

%% file: joa-environments-english.tex
\newtheorem{theoremABC}{Theorem}

\theoremstyle{definition}
\theoremstyle{remark}



%% file: joa-latex-shortcuts_1902_v1.tex
\newcommand{\N}{{\mathbb{N}}}

\newcommand{\R}{{\mathbb{R}}}

\newcommand{\Ll}{{\mathcal{L}}}   

\newcommand{\Ss}{{\mathcal{S}}}

\newcommand{\Ho}{{\rm H}}              

\newcommand{\lambdacan}{\lambda_{\rm can}} 

\newcommand{\Arnold}{{Arnol$'$d}}           
\newcommand{\eps}{{\varepsilon}}

\def\NABLA#1{{\mathop{\nabla\kern-.5ex\lower1ex\hbox{$#1$}}}}
\def\Nabla#1{\nabla\kern-.5ex{}_{#1}}
\def\Tabla#1{\Tilde\nabla\kern-.5ex{}_{#1}}
\def\abs#1{\mathopen|#1\mathclose|}

\renewcommand{\Tilde}{\widetilde}

\newlength\eqshift
\renewcommand\theequation{\thesection.\arabic{equation}}
\let\savetheequation\theequation

\makeatletter
\renewcommand*\env@matrix[1][\arraystretch]{%
  \edef\arraystretch{#1}%
  \hskip -\arraycolsep
  \let\@ifnextchar\new@ifnextchar
  \array{*\c@MaxMatrixCols c}}
\makeatother

\makeatletter
\let\save@mathaccent\mathaccent
\newcommand*\if@single[3]{%
  \setbox0\hbox{${\mathaccent"0362{#1}}^H$}%
  \setbox2\hbox{${\mathaccent"0362{\kern0pt#1}}^H$}%
  \ifdim\ht0=\ht2 #3\else #2\fi
  }
\newcommand*\rel@kern[1]{\kern#1\dimexpr\macc@kerna}
\newcommand*\widebar[1]{\@ifnextchar^{{\wide@bar{#1}{0}}}{\wide@bar{#1}{1}}}
\newcommand*\wide@bar[2]{\if@single{#1}{\wide@bar@{#1}{#2}{1}}{\wide@bar@{#1}{#2}{2}}}
\newcommand*\wide@bar@[3]{%
  \begingroup
  \def\mathaccent##1##2{%
    \let\mathaccent\save@mathaccent
    \if#32 \let\macc@nucleus\first@char \fi
    \setbox\z@\hbox{$\macc@style{\macc@nucleus}_{}$}%
    \setbox\tw@\hbox{$\macc@style{\macc@nucleus}{}_{}$}%
    \dimen@\wd\tw@
    \advance\dimen@-\wd\z@
    \divide\dimen@ 3
    \@tempdima\wd\tw@
    \advance\@tempdima-\scriptspace
    \divide\@tempdima 10
    \advance\dimen@-\@tempdima
    \ifdim\dimen@>\z@ \dimen@0pt\fi
    \rel@kern{0.6}\kern-\dimen@
    \if#31
      \overline{\rel@kern{-0.6}\kern\dimen@\macc@nucleus\rel@kern{0.4}\kern\dimen@}%
      \advance\dimen@0.4\dimexpr\macc@kerna
      \let\final@kern#2%
      \ifdim\dimen@<\z@ \let\final@kern1\fi
      \if\final@kern1 \kern-\dimen@\fi
    \else
      \overline{\rel@kern{-0.6}\kern\dimen@#1}%
    \fi
  }%
  \macc@depth\@ne
  \let\math@bgroup\@empty \let\math@egroup\macc@set@skewchar
  \mathsurround\z@ \frozen@everymath{\mathgroup\macc@group\relax}%
  \macc@set@skewchar\relax
  \let\mathaccentV\macc@nested@a
  \if#31
    \macc@nested@a\relax111{#1}%
  \else
    \def\gobble@till@marker##1\endmarker{}%
    \futurelet\first@char\gobble@till@marker#1\endmarker
    \ifcat\noexpand\first@char A\else
      \def\first@char{}%
    \fi
    \macc@nested@a\relax111{\first@char}%
  \fi
  \endgroup
}
\makeatother

\long\def\symbolfootnote[#1]#2{\begingroup%
\def\thefootnote{\fnsymbol{footnote}}\footnote[#1]{#2}\endgroup}

\newcommand{\LW}{L_{\rm W}}            
\newcommand{\HW}{H_{\rm W}}           
\newcommand{\Tflat}{T_{\rm flat}}

%% file: WEBER-ED-delay-arXiv.bbl
\begin{thebibliography}{WKTAW18}

\bibitem[AFS18a]{Albers:2018c}
Peter Albers, Urs Frauenfelder, and Felix Schlenk.
\newblock {A compactness result for non-local unregularized gradient flow
  lines}.
\newblock {\em ArXiv e-prints}, February 2018.
\newblock To appear in {\it Journal of Fixed Point Theory and Applications}.

\bibitem[AFS18b]{Albers:2018a}
Peter Albers, Urs Frauenfelder, and Felix Schlenk.
\newblock {An iterated graph construction and periodic orbits of Hamiltonian
  delay equations}.
\newblock {\em ArXiv e-prints}, February 2018.
\newblock To appear in {\it Journal of Differential Equations}.

\bibitem[AFS18c]{Albers:2018b}
Peter Albers, Urs Frauenfelder, and Felix Schlenk.
\newblock {What might a Hamiltonian delay equation be?}
\newblock {\em ArXiv e-prints}, February 2018.

\bibitem[Bel77]{Belbruno:1977a}
E.~A. Belbruno.
\newblock Two-body motion under the inverse square central force and equivalent
  geodesic flows.
\newblock {\em Celestial mechanics}, 15(4):467--476, Aug 1977.

\bibitem[Boh13]{Bohr:1913a}
Niels Bohr.
\newblock {On the Constitution of Atoms and Molecules, Part I}.
\newblock {\em Philosophical Magazine}, 26(151):1 -- 24, 1913.

\bibitem[Bus26]{Bush:1926a}
Vannevar Bush.
\newblock The force between moving charges.
\newblock {\em Journal of Mathematics and Physics}, 5(1-4):129--157, 1926.

\bibitem[Dar28]{Darwin:1928a}
Charles~Galton Darwin.
\newblock {The wave equations of the electron}.
\newblock {\em Proceedings of the Royal Society of London. Series A.}, 118,
  1928.

\bibitem[Ein17]{Einstein:1917a}
Albert Einstein.
\newblock {Zum Quantensatz von Sommerfeld und Epstein}.
\newblock {\em Deutsche Physikalische Gesellschaft, Verhandlungen}, 19:82 --
  92, 1917.

\bibitem[GG96]{Graneau:1996a}
Peter Graneau and Neal Graneau.
\newblock {\em Newtonian Electrodynamics}.
\newblock World Scientific Publishing Co., Inc., River Edge, NJ, 1996.

\bibitem[Gor28]{Gordon:1928a}
Walter Gordon.
\newblock {Die Energieniveaus des Wasserstoffatoms nach der Diracschen
  Quantentheorie des Elektrons}.
\newblock {\em Zeitschrift f{\"u}r Physik}, 48(1):11--14, Jan 1928.

\bibitem[Gou89]{Goursat:1889a}
E.~Goursat.
\newblock {Sur les transformations isogonales en M\'{e}canique}.
\newblock {\em Comptes Rendus des S\'{e}ances de l'Acad\'{e}mie des Sciences},
  108:446 -- 448, 1889.

\bibitem[Gut90]{Gutzwiller:1990a}
Martin~C. Gutzwiller.
\newblock {\em Chaos in classical and quantum mechanics}, volume~1 of {\em
  Interdisciplinary Applied Mathematics}.
\newblock Springer-Verlag, New York, 1990.

\bibitem[Hel47]{Helmholtz:1847a}
Hermann~von Helmholtz.
\newblock {\em {\"U}ber die Erhaltung der Kraft}.
\newblock Humboldt-Universit{\"a}t zu Berlin, 1847.

\bibitem[HWZ17]{Hofer:2017a}
Helmut {Hofer}, Kris Wysocki, and Eduard Zehnder.
\newblock {Polyfold and Fredholm Theory}. 714 pages, 
\newblock {\em ArXiv e-prints}, July 2017.

\bibitem[HZ11]{Hofer:2011a}
Helmut Hofer and Eduard Zehnder.
\newblock {\em Symplectic invariants and {H}amiltonian dynamics}.
\newblock Modern Birkh{\"a}user Classics. Birkh{\"a}user Verlag, Basel, 2011.
\newblock Reprint of the 1994 edition.

\bibitem[Kep03a]{Keppeler:2003a}
Stefan Keppeler.
\newblock Semiclassical quantisation rules for the {D}irac and {P}auli
  equations.
\newblock {\em Ann. Physics}, 304(1):40--71, 2003.

\bibitem[Kep03b]{Keppeler:2003b}
Stefan Keppeler.
\newblock {\em {Spinning Particles -- Semiclassics and Spectral Statistics}}.
\newblock Springer Berlin Heidelberg, 2003.

\bibitem[KTA94]{Koch-Torres-Assis:1994a}
Andr{\'e} Koch Torres~Assis.
\newblock {\em Weber's Electrodynamics}.
\newblock Springer, Dordrecht, 1994.

\bibitem[LC20]{Levi-Civita:1920a}
T.~Levi-Civita.
\newblock Sur la r{\'e}gularisation du probl{\`e}me des trois corps.
\newblock {\em Acta Math.}, 42:99--144, 1920.

\bibitem[LR47]{Lamb:1947a}
Willis~E. Lamb and Robert~C. Retherford.
\newblock Fine structure of the hydrogen atom by a microwave method.
\newblock {\em Phys. Rev.}, 72:241--243, Aug 1947.

\bibitem[Max55]{Maxwell:1855a}
James~Clerk Maxwell.
\newblock {On Faraday's lines of force}.
\newblock {\em Transactions of the Cambridge Philosophical Society},
  X(I.):155--229, 1855.

\bibitem[Max73]{Maxwell:1873a}
James~Clerk Maxwell.
\newblock {\em A treatise on electricity and magnetism}, volume~I.
\newblock Clarendon Press, 1873.

\bibitem[Mil83]{Milnor:1983a}
John Milnor.
\newblock On the geometry of the kepler problem.
\newblock {\em The American Mathematical Monthly}, 90(6):353--365, 1983.

\bibitem[Mos70]{Moser:1970a}
J{{\"u}}rgen Moser.
\newblock Regularization of {K}epler's problem and the averaging method on a
  manifold.
\newblock {\em Comm. Pure Appl. Math.}, 23:609--636, 1970.

\bibitem[MR82]{Mehra:1982a}
Jagdisch Mehra and Helmut Rechenberg.
\newblock {\em {The Historical Development of Quantum Theory}}, volume 1, part
  1 - {The Quantum Theory of Planck, Einstein, Bohr and Sommerfeld: Its
  Foundation and the Rise of its Difficulties 1900 -- 1925}.
\newblock Springer-Verlag, New York-Heidelberg-Berlin, 1982.

\bibitem[MR87]{Mehra:1987b}
Jagdisch Mehra and Helmut Rechenberg.
\newblock {\em {The Historical Development of Quantum Theory}}, volume 5 -
  {Erwin Schr\"odinger and the Rise of Wave Mechanics}, part 2 - {The Creation
  of Wave Mechanics; Early Response and Aplications 1925--1926}.
\newblock Springer-Verlag, New York-Heidelberg-Berlin, 1987.

\bibitem[MS17]{McDuff:2017b}
Dusa McDuff and Dietmar Salamon.
\newblock {\em Introduction to symplectic topology}.
\newblock Oxford Graduate Texts in Mathematics. Oxford University Press,
  Oxford, third edition, 2017.

\bibitem[Neu68a]{Neumann:1868a}
Carl Neumann.
\newblock {\em {Die Principien der Elektrodynamik: eine mathematische
  Untersuchung}}.
\newblock T{\"u}binger Universit{\"a}tsschriften a. d. J. 1868. H. Laupp, 1868.
\newblock Reprinted In Mathematischen Annalen, Vol. 17, pp. 400 - 434 (1880).

\bibitem[Neu68b]{Neumann:1868b}
Carolo Neumann.
\newblock Theoria nova phaenomenis electricis applicanda.
\newblock {\em Annali di Matematica Pura ed Applicata (1867-1897)},
  2(1):120--128, Aug 1868.

\bibitem[Neu96]{Neumann:1896a}
Carl Neumann.
\newblock {\em Allgemeine Untersuchungen {\"u}ber das Newton'sche Princip der
  Fernwirkungen: mit besonderer R{\"u}cksicht auf die elektrischen Wirkungen}.
\newblock B. G. Teubner, Leipzig, 1896.

\bibitem[Osi77]{Osipov:1977a}
Yu.~S. Osipov.
\newblock The kepler problem and geodesic flows in spaces of constant
  curvature.
\newblock {\em Celestial mechanics}, 16(2):191--208, Oct 1977.

\bibitem[Sch94]{Schweber:1994a}
Silvan Schweber.
\newblock {\em {QED and the Men Who Made It: Dyson, Feynman, Schwinger, and
  Tomonaga}}.
\newblock Princeton University Press, 1994.

\bibitem[Som16]{Sommerfeld:1916a}
Arnold Sommerfeld.
\newblock {Zur Quantentheorie der Spektrallinien}.
\newblock {\em Annalen der Physik}, 356(17):1--94, 1916.

\bibitem[Som21]{Sommerfeld:1921a}
Arnold Sommerfeld.
\newblock {\em {Atombau und Spektrallinien}}.
\newblock Friedrich Vieweg und Sohn (Braunschweig), {Zweite Auflage,} %
  edition, 1921.

\bibitem[TLKT67]{Thomson:1867a}
William Thomson (Lord~Kelvin) and Peter~Guthrie Tait.
\newblock {\em Treatise on natural philosophy. Vol. I}.
\newblock Clarendon press, Oxford, 1867.

\bibitem[TLKT71]{Thomson:1871a}
William Thomson (Lord~Kelvin) and Peter~Guthrie Tait.
\newblock {\em {Handbuch der Theoretischen Physik}}.
\newblock Erster Band. Erster Theil. Translated by H.\,Helmholtz and
  G.\,Wertheim. Friedrich Vieweg und Sohn (Braunschweig), 1871.

\bibitem[TLKT79]{Thomson-Lord-Kelvin:1879a}
William Thomson (Lord~Kelvin) and Peter~Guthrie Tait.
\newblock {\em Treatise on natural philosophy}, volume~{\rm I, part I}.
\newblock Cambridge University Press, Cambridge, new edition, 1879.

\bibitem[Web46]{Weber:1846a}
Wilhelm Weber.
\newblock {Elektrodynamische Maassbestimmungen {\"u}ber ein allgemeines
  Grundgesetz der elektrischen Wirkung}.
\newblock {\em Abhandlungen bei Begr{\"u}ndung der K{\"o}nigl. S{\"a}chs.
  Gesellschaft der Wissenschaften am Tage der zweihundertj{\"a}hrigen
  Geburtstagfeier Leibnizen's}, 1846.
\newblock Reprinted in Wilhelm Weber's Werke (Springer, Berlin, 1893), Vol. 3,
  pp. 25 - 214.

\bibitem[Web57]{Weber:1857a}
Wilhelm Weber.
\newblock {Elektrodynamische Maassbestimmungen insbesondere Zur\"uckf\"uhrung
  der Stromintensit\"ats-Messungen auf mechanisches Maass}.
\newblock {\em Abhandlungen der K{\"o}nigl. S{\"a}chs. Gesellschaft der
  Wissenschaften}, 3:221 -- 290, 1857.
\newblock Reprinted in Wilhelm Weber's Werke (Springer, Berlin, 1893), Vol. 3,
  pp. 609 - 676.

\bibitem[Web71]{Weber:1871a}
Wilhelm Weber.
\newblock {Elektrodynamische Maassbestimmungen insbesondere {\"u}ber das
  Princip der Erhaltung der Energie}.
\newblock {\em Abhandlungen der K{\"o}nigl. S{\"a}chs. Gesellschaft der
  Wissenschaften}, 10:1 -- 61, 1871.
\newblock Reprinted in Wilhelm Weber's Werke (Springer, Berlin, 1894), Vol. 4,
  pp. 247 - 299.

\bibitem[Wes89]{Wesley:1989a}
James~Paul Wesley.
\newblock {Evidence for Weber-Wesley Electrodynamics}.
\newblock In {\em {Proceedings Conference on Foundations of Mathematics and
  Physics}}, pages 289 -- 343, Perugia, Italy, 27-29 September 1989.

\bibitem[WKTAW11]{Wolfschmidt:2011a}
Gudrun Wolfschmidt, Andr{\'e} Koch Torres~Assis, and Karl~Heinrich Wiederkehr.
\newblock {\em {Weber's Planetary Model of the Atom}}.
\newblock Tredition, Hamburg, 2011.

\bibitem[WKTAW18]{Wolfschmidt:2018a}
Gudrun Wolfschmidt, Andr{\'e} Koch Torres~Assis, and Karl~Heinrich Wiederkehr.
\newblock {\em {Weber's Planeten-Modell des Atoms}}.
\newblock Apeiron, Montreal, 2018.

\bibitem[Z{\"o}l72]{Zollner:1872a}
Johann Karl~Friedrich Z{\"o}llner.
\newblock {\em {{\"U}ber die Natur der Cometen: Beitr{\"a}ge zur Geschichte und
  Theorie der Erkenntnis}}.
\newblock Engelmann (Leipzig), 1872.

\end{thebibliography}
